\documentclass[conference]{IEEEtran}
\usepackage{cite}
\usepackage{amsmath,amssymb,amsfonts,mathtools}
\usepackage{algorithmic}
\usepackage{graphicx}
\usepackage{textcomp}
\usepackage{xcolor}
\usepackage{dsfont}
\usepackage[flushleft]{threeparttable}
\usepackage{array,booktabs,makecell}
\usepackage{comment}
\usepackage{subcaption}
\usepackage{amsthm}
\usepackage{nicefrac}
\allowdisplaybreaks
\newcommand{\bsm}[1]{\boldsymbol{#1}}
\newtheorem{lemma}{Lemma}  
\newtheorem{theorem}{Theorem}  
\def\BibTeX{{\rm B\kern-.05em{\sc i\kern-.025em b}\kern-.08em
    T\kern-.1667em\lower.7ex\hbox{E}\kern-.125emX}}

\begin{document}

\title{Scalable Syndrome-based Neural Decoders for \\Bit-Interleaved Coded Modulations}

\author{\IEEEauthorblockN{Gastón De Boni Rovella\IEEEauthorrefmark{1}\,\IEEEauthorrefmark{2}, Meryem Benammar\IEEEauthorrefmark{2}, Tarik Benaddi\IEEEauthorrefmark{3} , Hugo Meric\IEEEauthorrefmark{4}}
\IEEEauthorrefmark{1}TéSA Laboratory, Toulouse, France \\
\IEEEauthorrefmark{2}ISAE-SUPAERO, Université de Toulouse, France \\
\IEEEauthorrefmark{3}Thales Alenia Space, Toulouse, France \\
\IEEEauthorrefmark{4}Centre National d'Études Spatiales, Toulouse, France \\
Email: \{gaston.de-boni-rovella, meryem.benammar\}@isae-supaero.fr\vspace{-0.4cm}
}

\maketitle

\begin{abstract}
 In this work, we introduce a framework that enables the use of Syndrome-Based Neural Decoders (SBND) for high-order Bit-Interleaved Coded Modulations (BICM). To this end, we extend the previous results on SBND, for which the validity is limited to Binary Phase-Shift Keying (BPSK), by means of a theoretical channel modeling of the bit Log-Likelihood Ratio (bit-LLR) induced outputs. We implement the proposed SBND system for two polar codes $(64,32)$ and $(128,64)$, using a Recurrent Neural Network (RNN) and a Transformer-based architecture. Both implementations are compared in Bit Error Rate (BER) performance and computational complexity.
\end{abstract}\vspace{-0.1cm}


\section{Introduction}
With the recent introduction of 5G and beyond technologies, the interest in fast and reliable communication systems has increased in an unprecedented manner. At the same time, the ever-growing computational power has given machine learning a crucial role in future communication systems as a fast and robust alternative to some classical physical layer solutions like channel demodulation and decoding.

Early works in channel decoding \cite{Gruber_2017, OShea_2017} promptly faced the \textit{curse of dimensionality}, where the space of valid codewords was too large for a common Deep Neural Network (DNN) to explore and \textit{learn}. To tackle this problem, two main scalable alternatives were proposed: model-based solutions, that directly exploit the structure of the code \cite{Nachmani_2016,Nachmani_2021, Xu_2017}, and model-free solutions, that do not depend on the code and allow the integration of more sophisticated machine learning techniques \cite{Bennatan_2018_arxiv, Choukroun_2022, Lugosch_2018}. Model-based solutions are often neural extensions of the Belief Propagation (BP) algorithm that help tackle the negative impact of short cycles in BP decoding. However, implementations for semi-dense or dense codes, such as BCH or Polar codes \cite{Arikan_2009}, remain subpar.


A model-free approach, which we denote as Syndrome-Based Neural Decoder (SBND), was introduced more recently by Bennatan \textit{et al.} \cite{Bennatan_2018_arxiv}, and has since been implemented using different deep learning techniques \cite{Choukroun_2022, Caciularu_2021, Lugosch_2018,DeBoni_2023}. The main idea behind the SBND is to produce a symmetric decoder that does not depend on the codeword and can thus be trained with a unique codeword.  Although this provides us with a very promising framework, these works rely extensively on the properties of Binary Phase-Shift Keying (BPSK) and can be easily extended to Quadrature Phase-Shift Keying (QPSK). However, in order to allow for practical implementations, SBND has to be extended to higher-order modulations such as $M$-Quadrature Amplitude Modulation (QAM) and $M$-Phase-Shift Keying (PSK) for arbitrary $M$. In this work, we propose a decoder that can be directly applied to such linear modulation techniques. More particularly, we focus on Bit-Interleaved Coded Modulations (BICM) \cite{Caire_1998, Alvarado_2008} which, unlike classical coded modulation schemes, present the advantages of allowing the usage of any Forward Error Correction (FEC) code designed for memoryless channels, and of being more robust to burst errors.

The main difference between higher-order modulations and BPSK/QPSK is that the decoder is not directly fed with the channel output, but rather, with bit Log-Likelihood Ratios (bit-LLR) produced by the soft demodulator. Hence, in order to design an SBND for the BICM, we first start by characterizing the channel induced by bit-LLRs for two common modulation schemes. Next, we propose an SBND that extends that of \cite{Choukroun_2022, Caciularu_2021, Lugosch_2018,DeBoni_2023} to the case of high-order BICM. Finally, we analyze the performance of two main architectures for the neural-based decoders, namely, RNN-based and transformer-based, and compare their respective complexities. Finally, Section \ref{sec:conclusion} concludes the work.

The remainder of this work is organized as follows. Section \ref{sec:preliminary} introduces the system model and some preliminaries on channel modeling for BICM. In Section \ref{sec:decoder}, we describe the proposed decoding framework. Then, Section \ref{sec:experiments} presents two possible implementations for the deep learning-based portion of the decoder, compares experimentally their BER performances and analyzes their complexity. 
 
\noindent \textit{Notations:} 
Capital italic letters (e.g. $X$ and $\bsm{X}$) represent random variables and vectors whereas Roman and bold letters (e.g. $x$ and $\bsm{x}$) denote their respective realizations. Matrices are represented by non-italic capital letters (e.g. $\mathrm{H}$) and $\mathrm{I}_n$ denotes the $n\times n$ identity matrix. The Hadamard product between vectors is represented by $(.)$.  Sets are denoted by calligraphic letters $\mathcal{X}$, and for finite sets, $|\mathcal{X}|$ denotes the cardinality. $P_X(x)$ (resp. $P_{X|Y}(x|y)$) represents the probability distribution (resp. conditional) evaluated in $x$ (resp. $(x,y)$). Markov chains are denoted $X \leftrightarrow Y \leftrightarrow Z $ to mean $P_{Z|Y,X} = P_{Z|Y}$. $\mathds{P}$ (resp. $ \mathds{1}$) denotes the generic probability (resp. indicator) of an event. $[1:n]$ denotes the set of integers from $1$ to $n$ and $\lfloor \cdot \rfloor$ represents the floor function. The xor operation is noted $\oplus$.

\section{System model and preliminaries}\label{sec:preliminary}

\subsection{Bit-Interleaved Coded Modulations (BICM)}\label{sec:system_model}

\begin{figure}[htbp]
    \centerline{\includegraphics[width=\linewidth]{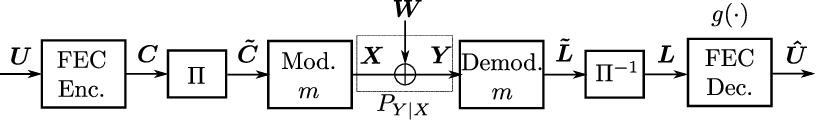}}
    \caption{General system model.}
    \label{fig:sys_model}
\end{figure}

Let us consider a BICM communication setting as depicted in Figure \ref{fig:sys_model}. In such a setting, an input binary message $\bsm{U}$ of $k$ bits, assumed to be independent and uniformly distributed, is first mapped through an $(n,k)$ linear block FEC encoder into a codeword $\bsm{C}$ of $n$ bits. Let us denote the parity check matrix of this code as $\mathrm{H}$. Then, a perfectly random interleaver $\Pi$, assumed to be known to the receiver as well, shuffles the bits of $\bsm{C}$ into an $n$-bit sequence $\bsm{\tilde{C}}$, which is then passed through a modulator. The complex-valued constellation is denoted by $\mathcal{X}$ and is assumed to be of order $m$ (i.e. $M=2^m$ states). The sequence of $n^\prime = n/m$ symbols $\bsm{X}$ is then fed to an Additive White Gaussian Noise (AGWN). Its output can be modeled as
    \begin{equation}
        \bsm{Y} = \bsm{X} + \bsm{W}, \ \text{s.t.} \  \bsm{W} \overset{\text{i.i.d}}{\sim} \mathcal{CN}(\bsm{0}, \sigma^2 \mathrm{I}_{n^\prime}) .
    \end{equation} 
  The demodulator, upon reception of the channel output $\bsm{Y}$, computes the bit-Log Likelihood Ratios (bit-LLR) associated with each bit $C$ given the corresponding received signal $Y$. The resulting bit-LLRs $\bsm{\tilde{L}}$ are then de-interleaved back to the original bit order. Based on the obtained bit-LLRs $\bsm{L}$, the FEC decoder $g(\cdot)$ produces an estimate $\bsm{\hat{u}}$ of the $k$ transmitted bits given by  $\bsm{\hat{U}} = g(\bsm{L})$.   

In this work, we investigate the design of decoders that minimize the Bit-Error Probability (BEP) defined by 
 \begin{equation}
     P_e(g) \triangleq \dfrac{1}{k}\sum_{i=1}^k \mathds{P}(\hat{U}_i \neq U_i). 
 \end{equation} 

\subsection{BICM equivalent channel models}
The main advantage of the BICM communication setting, as opposed to standard (non-interleaved) coded modulation schemes, is that a perfectly random interleaver maps every bit in $\bsm{C}$ evenly to each one of the $m$ bit positions in the constellation mapping. Hence, throughout the transmission, all bits in the sequence $\bsm{C}$ end up experiencing the same channel, which is averaged out over all bit positions.

An equivalent channel model for BICM was introduced in \cite{Caire_1998} and is given in Figure \ref{fig:BICM_channel_model}, where $P^s_{Y|C}$ denotes the effective channel distribution experienced by a bit transmitted over a position $s$ in the constellation mapping, $f^s(\cdot)$ denotes the bit-LLRs function (depending on the bit position) and is given for all $y \in \mathds{C}$ and $s \in [1:m]$ by  
\begin{equation}
    f^s(y) \triangleq \log\left( P^s_{Y|C}(y|0) \right)- \log\left(  P^s_{Y|C}(y|1)  \right),
\end{equation}
and the last operation is simply a decomposition into a hard decision $L^b \triangleq \mathds{1}(L<0)$ and a reliability measure $|L|$. 

\begin{figure}[htbp]
    \centerline{\includegraphics[width=0.77\linewidth]{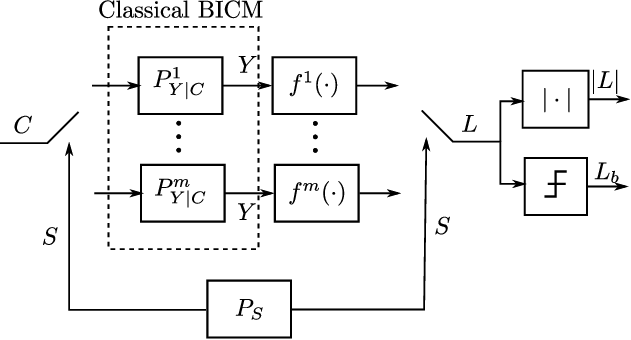}}
    \caption{BICM channel model extended to bit-LLRs.}\vspace{-0.4cm}
    \label{fig:BICM_channel_model}
\end{figure}
 
In order to motivate our proposed decoder of Section \ref{sec:decoder}, we need to characterize three channel distributions $P_{L^b|C}$. To this end, we will recall the results on $P_{Y|C}$ and $ P_{L|C}$, and extend them to model the channel $P_{L^b|C}$.

\subsubsection{Channel distributions $P_{Y|C}$ and $P_{L|C}$} Let us recall the equivalent channel models as of the BICM in Figure \ref{fig:BICM_channel_model}.

\begin{lemma}[BICM classical channel models \cite{Caire_1998, Alvarado_2008}]
An equivalent channel distribution of a classical BICM setting is given for all $\bsm{y} \in \mathds{C}^{n^\prime}$, $\bsm{l} \in \mathds{R}^n$, and $\bsm{c} \in \{0,1\}^n$ by \vspace{-0.15cm}
 \begin{equation*}
       P_{\bsm{Y}|\bsm{C}}(\bsm{y}|\bsm{c} ) \text{=} \prod_{i=1}^n P_{Y|C}(y_{\lfloor \nicefrac{i}{m} \rfloor }|c_i)   \ , \  P_{\bsm{L}|\bsm{C}}(\bsm{l}|\bsm{c} ) \text{=}  \prod_{i=1}^n P_{L|C}(l_i|c_i)
   \end{equation*}
    where, the distributions $P_{Y|C}$ and$P_{L|C}$ are given by\vspace{-0.15cm}
    \begin{IEEEeqnarray}{rCl}
    P_{Y|C}(y|c) &=& \dfrac{1}{m|\mathcal{X}_{c}^s|} \sum_{s=1}^m  \sum_{x \in \mathcal{X}_{c}^s} P_{Y|X} (y|x), \\
     P_{L|C}(l|c) &=& \dfrac{1}{m|\mathcal{X}_{c}^s|} \sum_{s=1}^m  \sum_{x \in \mathcal{X}_{c}^s} P^s_{L|X} (l|x),
\end{IEEEeqnarray}
and $\mathcal{X}_{c}^s$ is the set of symbols for which the $s$-th bit equals $c$.
  \end{lemma}
    \begin{proof}
The proofs can be found in \cite[Section 3.4]{Alvarado_2008}. 
\end{proof}
Note that the distributions $P^s_{L|X}$, for $s \in [1:m]$, can be easily obtained for BPSK and QPSK modulations; however, they are very challenging to obtain analytically for higher-order more generic modulation schemes. In this work, we do not seek to characterize these distributions in a closed form, but rather, will use them to derive an equivalent binary channel model for the hard decisions on the bit-LLRs. 

\subsubsection{Bit-LLRs binary channel distribution $P_{L^b|C}$} In order to motivate the structure of the decoder developed in this work, we need to characterize the channel distribution $P_{L^b|C}$. To this end, let us consider, for instance, the 16-QAM and 8-PSK constellations with Gray labeling shown in Figures \ref{fig:8-PSK-regions} and \ref{fig:16-QAM-regions}. 
 
\begin{theorem}[Bit-LLRs binary channel model]\label{thm:binary_channel_model}
   For all $\bsm{l}^b , \bsm{c} \in \{0,1\}^n$, the following equation holds: 
   \begin{equation}
       P_{\bsm{L}^b|\bsm{C}}(\bsm{l}^b|\bsm{c} ) = \prod_{i=1}^n P_{L^b|C}(l^b_i|c_i).
   \end{equation}
   Besides, for the 8-PSK and 16-QAM under Gray labeling, the channel $P_{L^b|C}$ can be approximated by
   \begin{equation}\label{eq:equivalent_BSC}
       \bsm{L^b} = \bsm{C} \oplus \bsm{W}^b \  \text{  s.t.  } \ \bsm{W}^b \overset{\text{i.i.d}}{\sim} \text{Bern}(q),
   \end{equation}
   where $\bsm{W}^b $ is independent of $\bsm{C}$ and $q \triangleq \dfrac{1}{m} \displaystyle \sum_{s=1}^m P^s_{L^b|C}(1|0) $. 
\end{theorem}
\begin{proof}
    The proof is relegated to Appendix \ref{app:proof_binary_channel_model}. 
\end{proof}
The present Theorem states that for the 8-PSK constellation under Gray labeling, the channel $P_{L^b |C}$ is in fact a memoryless stationary Binary Symmetric Channel (BSC), as described in \eqref{eq:equivalent_BSC}. As for the 16-QAM channel, we show that this channel model is valid for a wide range of Signal-to-Noise Ratios (SNR), i.e., except for very low SNR. Although we prove this theorem for only these two constellations, which are the most common in practice, the result can be extended to M-PSK and M-QAM constellations following the same lines of the proof.  

\section{Proposed system: SBND for BICM}\label{sec:decoder}
In what follows, we start by stating a result on Maximum A Posteriori (MAP) decoding, then we review existing literature on SBND for the BPSK and QPSK modulations and introduce our proposed solution.
\subsection{Decoding as a binary noise detection problem}\label{sec:cw_to_mess}
The decoder architecture introduced in this section stems from a first result on the equivalence between decoding and denoising.
First, let us recall that, given the channel model in Figure \ref{fig:sys_model}, the optimal decoding rule $g(\cdot)$ is the MAP rule, and is given by 
\begin{equation}
    \bsm{\hat{u}} = g(\bsm{l}) \triangleq \underset{\bsm{u} \in \{0,1\}^k}{\textrm{argmax}} P_{\bsm{U}|\bsm{L}} (\bsm{u}|\bsm{l}). 
\end{equation}
Let us consider the result of Theorem  \ref{thm:binary_channel_model} in \eqref{eq:equivalent_BSC}, and let us define the pseudo-inverse function of the FEC encoder as $p_{inv}(\cdot)$, i.e., $\bsm{u} = p_{inv}(\bsm{c})$. 
Then exploiting the linearity of $p_{inv}(\cdot)$, we can show that there exists a binary noise sequence $\bsm{W}^b_u$ independent from $\bsm{U}$ such that 
\begin{equation}
    p_{inv}(\bsm{L}^b) = \bsm{U} \oplus  \bsm{W}^b_u.
\end{equation}
Hence, we rewrite the MAP decoding rule as
\begin{equation}
   \underset{\bsm{u} \in \{0,1\}^k}{\textrm{argmax}} P_{\bsm{U}|\bsm{L}} (\bsm{u}|\bsm{l}) =  p_{inv}(\bsm{l}^b)  \oplus \underset{\bsm{w}  \in \{0,1\}^k}{\textrm{argmax}} P_{\bsm{W}^b_u|\bsm{L}} (\bsm{w}|\bsm{l}) .
\end{equation}
Hence, MAP decoding of $\bsm{U}$ amounts to MAP detection of the binary noise --or \textit{bit-flip}-- sequence $\bsm{W}^b_u$. 

\subsection{Previous works on SBND for BPSK and QPSK modulations}\label{sec:SBND-BPSK}
Previous work from Bennatan \textit{et al.} \cite{Bennatan_2018_arxiv} proved that, under the assumption of a BPSK modulation and an AWGN channel, there exists a noise sequence $\bsm{W}$ such that 
\begin{eqnarray}
    \bsm{Y} = \bsm{X}.\bsm{W}  \text{ and } \bsm{Y}^b = \bsm{C}\oplus \bsm{W}^b,
\end{eqnarray}
where $\bsm{Y}^b$ and $\bsm{W}^b$ denote the binary hard decisions of $\bsm{Y}$ and $\bsm{W}$. 
From this, they showed that the knowledge of the received signal's module $|\bsm{Y}|$ and the hard-decision syndrome $\mathrm{H}\bsm{Y}^b$ is enough to estimate the codeword binary noise $\bsm{W}^b$, i.e., 
\begin{equation}
    P_{\bsm{W}^b|\bsm{Y}} (\bsm{w}^b|\bsm{y})  =  P_{\bsm{W}^b| \, |\bsm{Y}|, \mathrm{H}\bsm{Y}^b} (\bsm{w}^b  | \,|\bsm{y}|, \mathrm{H}\bsm{y}^b).
\end{equation} 
Moreover, since $|\bsm{Y}| = |\bsm{W}|$ and  $ \mathrm{H}\bsm{Y}^b = \mathrm{H}\bsm{W}^b$, the posterior distribution $  P_{\bsm{W}^b|\bsm{Y}} $ does not depend on the transmitted codeword $\bsm{C}$. Hence, we can train a neural network to approximate this posterior using only one codeword, as long as the noise $\bsm{W}$ remains random.

In a previous work from the authors \cite{DeBoni_2023}, the result of \cite{Bennatan_2018_arxiv} was improved by directly estimating the bit-flips on the information bits rather than on the codewords, for both systematic and non-systematic codes. To this end, we showed that there exists a binary noise sequence $ \bsm{W}^b_u$ such that 
\begin{equation}
    p_{inv}(\bsm{Y}^b) = \bsm{U} \oplus  \bsm{W}^b_u , 
\end{equation}
and proved that the posterior can be written as
\begin{equation}
    P_{\bsm{W}^b_u|\bsm{Y}} (\bsm{w}_u^b|\bsm{y})  =  P_{\bsm{W}^b_u| \, |\bsm{Y}|, \mathrm{H}\bsm{Y}^b} (\bsm{w}_u^b| \, |\bsm{y}|,\mathrm{H}\bsm{y}^b).
\end{equation}
and is also independent of the message $\bsm{U}$, which still allows single codeword training. 

The extension of all these results to QPSK modulation is straightforward. In this work, we seek to generalize syndrome-based decoding to arbitrary higher-order modulations when the decoder is fed with bit-LLRs rather than the channel output.
 
 \subsection{Proposed solution: SBND for BICM}\label{sec:SBND-BICM}
In the following, we build on the results of \cite{DeBoni_2023} and on Theorem \ref{thm:binary_channel_model} to suggest an SBND for higher-order modulations. We start by proving that $|\bsm{L}|$ and $\mathrm{H}\bsm{L}^b$ are sufficient statistics for the detection of $\bsm{W}^b_u$. 

\begin{theorem}[Sufficient statistics]\label{thm:sufficient_statistic} Considering the problem setting and the result of Theorem \ref{thm:binary_channel_model}, we have that 
 \begin{equation}
    P_{\bsm{W}^b_u|\bsm{L}} (\bsm{w}_u^b|\bsm{l})  =  P_{\bsm{W}^b_u| \, |\bsm{L}|, \mathrm{H}\bsm{L}^b} (\bsm{w}_u^b| \, |\bsm{l}|, \mathrm{H}\bsm{l}^b).
\end{equation}  
\end{theorem}
\begin{proof}
    Proof is relegated to Appendix \ref{app:proof_sufficient_statistics}
\end{proof}
This previous theorem allows us to state that $|\bsm{L}|, \mathrm{H}\bsm{L}^b$ are sufficient to compute the posterior distribution of $\bsm{W}^b_u$. Hence, we suggest an implementation of the SBND for higher-order modulations as shown in Figure \ref{fig:SBND_modulations}. Observe that, as opposed to the BPSK and QPSK cases, the reliabilities $|\bsm{L}|$ are not independent of the transmitted symbols, and thus, one cannot train using only one codeword ($\bsm{c}= \bsm{0}$ for instance). This imposes training over randomly generated codewords, in order to ensure variability of the transmitted symbols.

\begin{figure}[htbp]
    \centering
   \includegraphics[width=\linewidth]{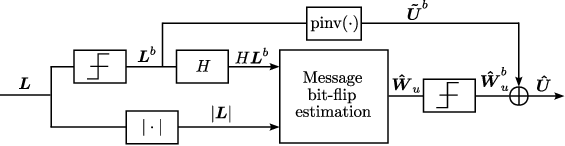}
    \caption{Suggested SBND for higher-order modulations}
    \label{fig:SBND_modulations}
\end{figure} 

\section{Experiments}\label{sec:experiments}
In this section, we briefly introduce two possible implementations of the bit-flip estimator of Section \ref{sec:SBND-BICM}, test the proposed decoder with both implementations, and compare both their performances and computational complexities. 

\subsection{Possible architectures and training}
Several DNN-based architectures can be considered for the bit-flip estimator. Due to their favorable performances and substantial differences in terms of complexity, the RNN estimator of \cite{Bennatan_2018_arxiv} and the transformer-based estimator of \cite{Choukroun_2022} were implemented, employing the message-wise approach of \cite{DeBoni_2023}. The main difference between these architectures lies in the number of weights that comprise each network and the number of operations needed to run each solution (see Section \ref{sec:complexity}). For a detailed description of these two solutions, refer to \cite{Bennatan_2018_arxiv, Choukroun_2022, DeBoni_2023}. 

Regarding the RNN-based estimator, it is a Gated Recurrent Unit (GRU)-based RNN in which each GRU cell \cite{Cho_2014} is composed of $\alpha(2n-k)$ GRU units, with $\alpha$ an arbitrary scaling parameter. The RNN consists of $d_l$ layers --i.e. $d_l$ stacked GRU cells-- and each unit performs $T$ time steps, with a final dense layer with a linear activation that outputs $\bsm{\hat{w}}_u$ of size $k$. Training is carried out with a batch size of $2^{12}$. 

As for the transformer architecture, it is determined by three hyperparameters: the embedding dimension $d_e$, the number of heads $d_h$ in the multi-head attention mechanism, and the number of encoder layers $N$ that are connected before the output dense layers. The large number of operations performed in the forward pass of the transformer architecture imposes a significantly smaller batch size, which is set to $2^8$. 

Both architectures have a final output dense layer with a linear activation that produces $\bsm{\hat{w}}_u$, which is then thresholded to obtain $\bsm{\hat{w}}^b_u = \mathds{1}(\bsm{\hat{w}}_u > 0)$. Training and testing were carried out using Google's TensorFlow library \cite{MartinAbadi_2015} and the Keras API \cite{Chollet_2015}, using the Adam optimizer \cite{Kingma_2014} with a learning rate of $\mu = 10^{-3}$ and a binary cross-entropy loss function. For both systems, codewords are generated using an AWGN of normalized SNR $E_b/N_0=5$dB. Important parameters are reported in Table \ref{tab:parameters}.


\renewcommand{\arraystretch}{1.5}
\begin{table}
\vspace*{0.2cm}
    \centering
    \begin{threeparttable}
        \begin{tabular}{c|c|c|c|c}
        \hline\hline
        RNN  & $\alpha=5$ & $T=5$ & $d_l=5$ & batch size $=2^{12}$  \\  
        \hline
        Transformer  &   $d_e=128$  &   $d_h=8$ & $N=10$ & batch size $=2^8$ \\
        \hline\hline 
        \end{tabular}
        \caption{Model and training parameters.}\vspace{-0.5cm}
        \label{tab:parameters}
    \end{threeparttable}
\end{table}

\subsection{Results}
Both of the considered architectures were applied to the decoding of two rate-$\nicefrac{1}{2}$ polar codes, namely the $(128, 64) $ and $(64,32)$ polar codes\footnote{The parity-check matrices were taken from the channel code database in https://rptu.de/en/channel-codes.}. An Ordered Statistics Decoding (OSD) algorithm is also added as a near-optimal decoding benchmark, along with an ML bound that records an error only when the OSD encounters a decoding failure \emph{and} the obtained codeword has a higher probability than the transmitted one.

The results are displayed in Figure \ref{fig:polar_codes}. Regarding the $(64,32)$ polar code, the RNN-based decoder presents a decoding performance that is very close to the OSD, and surpasses the transformer architecture for all the considered $E_b/N_0$. A similar conclusion can be drawn for the $(128,64)$ polar code, except this time, the gap between the neural-based and near-optimal solutions is more significant, especially for low-to-medium $E_b/N_0$ regions. It is worth mentioning that for the $(128,64)$ polar code, during training, both solutions have seen at most $10^{-8} \%$ of valid codewords, proving that the models are very much able to learn a proper decoding rule only by seeing a very small fraction of the data.

In the next section, we analyze the complexity involved in each solution and compare them accordingly.
\begin{figure*}
     \centering
     \begin{subfigure}[b]{0.49\textwidth}
        \centering
        \includegraphics[width=0.95\linewidth]{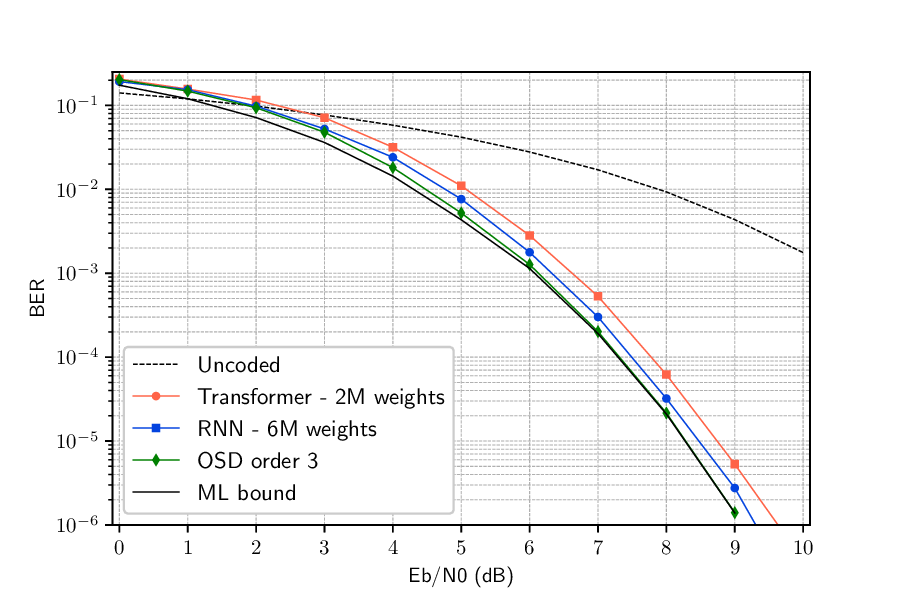}
        \label{fig:BER_128}
     \end{subfigure}
     \hfill
     \begin{subfigure}[b]{0.49\textwidth}
     \centering
        \includegraphics[width=0.95\linewidth]{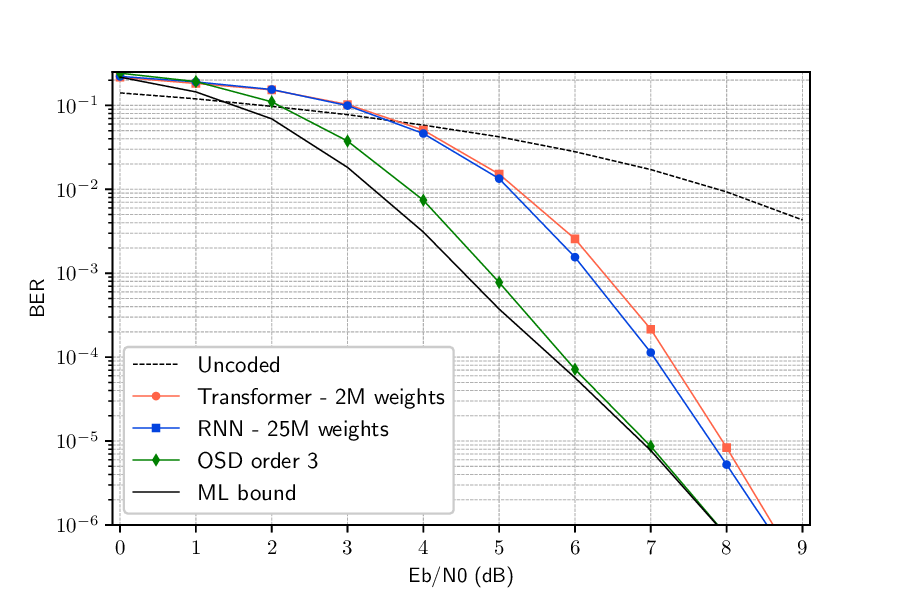}
        \label{fig:BER_64}
     \end{subfigure}
        \caption{Error rate studies for two rate $1/2$ polar codes:  $(64,32)$ (left) and $(128,64)$ (right).}\vspace{-0.5cm}
        \label{fig:polar_codes}
\end{figure*}

\subsection{Complexity analysis}\label{sec:complexity}
In this section, we study the complexity of each of the RNN-based and Transformer-based implementations of the decoder, both at the training and inference.

As shown in \cite{Dey_2017}, the number of parameters (or \textit{weights}) in a GRU unit --including biases, which were not included in the original work \cite{Cho_2014}-- is equal to $3(n_o^2 + n_in_o + n_o)$, where $n_i$ and $n_o$ depict the number of inputs and outputs, respectively. In the case of a dense layer, the number of weights is given by $n_in_o + n_o$. Applying these results to the RNN architecture gives the expression for the total number of weights of the RNN-based decoder for an $(n,k)$ linear code:
\begin{equation}
    \mathds{W}_{\text{RNN}}= 3(2d_l-1)\alpha^2r^2 + 3(r+d_l+\nicefrac{k}{3})\alpha r +k ,
\end{equation}
where $r\triangleq 2n-k$ is the size of the input vector given by $[\,|\bsm{l}|,\textrm{H}\bsm{l}^b \,]$. For the codes and hyperparameters selected, the following approximation holds to within a $0.5\%$ error margin,
\begin{equation}
    \mathds{W}_{\text{RNN}} \approx 3((2d_l-1)\alpha^2 + \alpha )r^2,
\end{equation}
yielding a total number of weights that increases as $\mathcal{O}(r^2)$ for a fixed network depth $d_l$ and scaling factor $\alpha$. 

Regarding the transformer architecture, the number of weights is computed with respect to the embedding dimension $d_e$, the input size $r$, and the number of encoder layers $N$:
\begin{equation}
    \mathds{W}_{\text{T}} = 12Nd_e^2 + (13N+r+3)d_e + (r+1)k + 1.
\end{equation}
Similarly to the previous case, an approximation can be made to within an error margin of $2\%$ for both studied codes:
\begin{equation}
    \mathds{W}_{\text{T}} \approx 12Nd_e^2 + 13Nd_e.
\end{equation}
Let us observe that, for the block sizes considered, the dominant terms do not contain the code parameters $(n,k)$, producing a network structure that does not depend on the code size as strongly as the RNN architecture. Observe, however, that there is a hidden dependence on the block length if we consider that larger codes may require larger embedding spaces and potentially more encoder layers. These values were kept the same throughout both implementations and hence, the number of parameters was almost exactly the same.

The final number of weights for each decoder is included in Fig. \ref{fig:polar_codes}. The transformer solution has a clear advantage in that the number of weights remains essentially the same as the input size increases, due to the embedding layer of fixed dimension $d_e\!=\!128$. However, the RNN is much more shallow than the transformer: it is composed of $d_l\!=\!5$ recurrent layers, whereas the transformer has the embedding layer, $N\!=\!10$ encoders --each consisting of two batch normalization layers, an attention layer and a dense layer-- and two final dense layers with batch normalization. Additionally, as shown in \cite{Vaswani_2017}, each self-attention mechanism entails a computational complexity of $\mathcal{O}(r^2 d_e)$, which significantly increases the inference time with respect to the RNN. Over several tests, the RNN decoded approximately $8$ times faster than the transformer for the $(64,32)$ polar code and $25$ times faster for the $(128,64)$ code\footnote{It must be factored in that empirical values are deeply reliant on the software implementation, hardware characteristics, and its parallel computing capabilities. Nonetheless, the obtained decoding latencies are consistent with the expected results.}.

\section{Conclusion}\label{sec:conclusion}
In this work, we have introduced a complete framework for decoding linear block codes under a BICM setting. For this purpose, we have deduced an equivalent bit-LLR binary channel model for the 8-PSK and 16-QAM modulation techniques. Next, an SBND for higher-order modulation is proposed, which takes as input the bit-LLRs instead of the BPSK-modulated signal as previous works \cite{Bennatan_2018_arxiv, Choukroun_2022}. Finally, two possible architectures were implemented and compared, both in terms of performance and computational complexity.

\appendices

\section{Proof of Theorem \ref{thm:binary_channel_model}} \label{app:proof_binary_channel_model}
First, note that since that each $L^b$ is a deterministic function of $L$, and given that $ P_{\bsm{L}|\bsm{C}}$ is a memoryless channel, then so is $P_{\bsm{L}^b|\bsm{C}}$, i.e., for all $\bsm{l}_b, \bsm{c} \in \{0,1\}^n$, $ P_{\bsm{L}^b|\bsm{C}}(\bsm{l}^b|\bsm{c}) = \prod_{i=1}^n P_{L^b|C}(l^b_i|c_i).$
  
  To proceed with the proof, we will first prove that $P_{L^b|C}(0|0) = P_{L^b|C}(1|1)$. Consider the following result
     \begin{IEEEeqnarray}{rCl}
      &&   P_{L^b|C}(0|0)  = \dfrac{1}{m }\sum_{s=1}^m  P_{L^b|C,S} (0|0,s)   \\
     &\overset{(a)}{=}& \dfrac{1}{m| \mathcal{X}^s_0| }\sum_{s=1}^m \sum_{x \in \mathcal{X}^s_0}  P_{L^b|X,S} (0|x,s) \\
     &=& \dfrac{1}{m | \mathcal{X}^s_0|}\sum_{s=1}^m \sum_{x \in \mathcal{X}^s_0} \int_{\mathds{C}}  P_{L^b,Y|X,S} (0,y|x,s) dy \\
      &\overset{(b)}{=}&  \dfrac{1}{m | \mathcal{X}^s_0|}\sum_{s=1}^m \sum_{x \in \mathcal{X}^s_0} \int_{\mathcal{Y}_0^s} P_{Y|X,S} (y|x,s) dy \\
      &\overset{(c)}{=}& \dfrac{1}{m | \mathcal{X}^s_0|}\sum_{s=1}^m \sum_{x \in \mathcal{X}^s_0} \int_{\mathcal{Y}_0^s} P_{Y|X} (y|x) dy,
     \end{IEEEeqnarray} 
     where we define the decision region $\mathcal{Y}_0^s \triangleq \{ y \in \mathds{C},  f^s(y) > 0\}$, and where $(a)$ follows the unformity of the constellation, while $(b)$ follows from the Markov chain $C \leftrightarrow (X,S) \leftrightarrow L^b$, and $(c)$ from the Markov chain  $S \leftrightarrow  X \leftrightarrow Y$. 

    Next, we use a common simplification of the bit-LLRs, commonly known as \emph{approximate LLRs}, to write that:
     \begin{IEEEeqnarray}{rCl}
         l & = & \log\biggl(  \sum_{x \in \mathcal{X}_{0}^s} P_{Y|X} (y|x) \biggr) \!-\! \log\biggl(   \sum_{x \in \mathcal{X}_{1}^s} P_{Y|X} (y|x)) \biggr) \\
          &\approx& \log\left( \max_{x \in \mathcal{X}_{0}^s} P_{Y|X} (y|x) \right)   - \log\left(\max_{x \in \mathcal{X}_{1}^s} P_{Y|X} (y|x) \right) \\
          &=& \dfrac{1}{\sigma^2} \bigl( \min_{x \in \mathcal{X}_{1}^s} || y - x ||^2 - \min_{x \in \mathcal{X}_{0}^s} || y - x ||^2 \bigr) .
     \end{IEEEeqnarray}
        Hence, we can write that for all $ y \in \mathcal{C}$ and all $s$, 
        \begin{equation}
            l^b = \mathds{1} \bigl( \min_{x \in \mathcal{X}_{0}^s} || y - x ||^2 \geq  \min_{x \in \mathcal{X}_{1}^s} || y - x ||^2 \bigr)\label{eq:approx_LLR} . 
        \end{equation}
        
     Let us then consider the 8-PSK constellation with Gray mapping of Figure \ref{fig:8-PSK-regions}. Taking into account the simplification in \eqref{eq:approx_LLR}, we give in Figure \ref{fig:8-PSK-regions} the decision regions $\mathcal{Y}_0^s $ for $s \in [1:3]$. 
     \begin{figure}[htbp]
         \centering
         \includegraphics[width=0.9\linewidth]{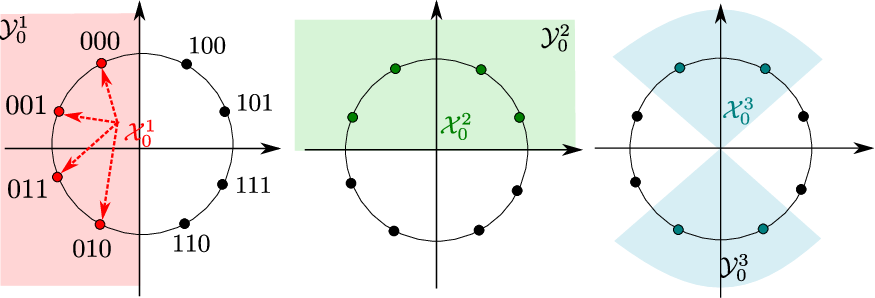}
         \caption{Decision regions of the 8-PSK constellation}
         \label{fig:8-PSK-regions}
     \end{figure}
    Let us now consider the case $s=1$. We have that $\mathcal{X}^1_0 = (\mathcal{X}^1_1)^\star$ and  $\mathcal{Y}^1_0 = (\mathcal{Y}^1_1)^\star$. Thus, by exploiting one of the symmetries of the normal Gaussian probability distribution $ P_{Y|X} (y|x) =  P_{Y|X} (y^\star|x^\star)$, one can easily prove that
    \begin{IEEEeqnarray}{rCl}\label{eq:symmetry}
 \sum_{x \in \mathcal{X}^1_0} \int_{\mathcal{Y}_0^1} P_{Y|X} (y|x) dy = \sum_{x \in \mathcal{X}^1_1} \int_{\mathcal{Y}_1^1} P_{Y|X} (y|x) dy . 
\end{IEEEeqnarray} 
Similar results can be proved for $s=2$ by noticing that $\mathcal{X}^2_0 = -  \mathcal{X}^2_1$,$\mathcal{Y}^2_0 = -  \mathcal{Y}^2_1$ and using the property  $P_{Y|X} (-y|-x)$. Finally, for $s=3$, note that $\mathcal{X}^3_0 =  \mathcal{X}^3_1 e^{j \pi/2}$ and $\mathcal{Y}^3_0 =  \mathcal{Y}^3_1 e^{j \pi/2}$ along with the symmetry $P_{Y|X} (y e^{j\phi}|x e^{j\phi}) = P_{Y|X} (y^\star|x^\star)$ for all $\phi$ yields the same result as \eqref{eq:symmetry}.
 Hence, recalling that $ | \mathcal{X}^s_0|=  | \mathcal{X}^s_1| $, it follows that 
    \begin{IEEEeqnarray}{rCl}
       P_{L^b|C}(0|0)  &= &  \dfrac{1}{m | \mathcal{X}^s_0|}\sum_{s=1}^m \sum_{x \in \mathcal{X}^s_0} \int_{\mathcal{Y}_0^s} P_{Y|X} (y|x) dy \\
      &=&   \dfrac{1}{m |\mathcal{X}^s_1|}\sum_{s=1}^m \sum_{x \in \mathcal{X}^s_1} \int_{\mathcal{Y}_1^s} P_{Y|X} (y|x) dy\\
      &=& P_{L^b|C}(1|1).   
     \end{IEEEeqnarray}  
       Hence, for the 8-PSK, the channel $P_{L^b|C}$ is a binary symmetric channel, with crossover probability $q$ given by 
       \begin{equation}
           q = P_{L^b|C}(1|0) = \dfrac{1}{m} \displaystyle\sum_{s=1}^m P^s_{L^b|C} (1|0).
       \end{equation} 

       Concerning the 16 QAM, we represent in Figure \ref{fig:16-QAM-regions} the decision regions $\mathcal{Y}_1^s $ for $s \in [1:4]$. 
           \begin{figure}[htbp]
         \centering
         \includegraphics[width=\linewidth]{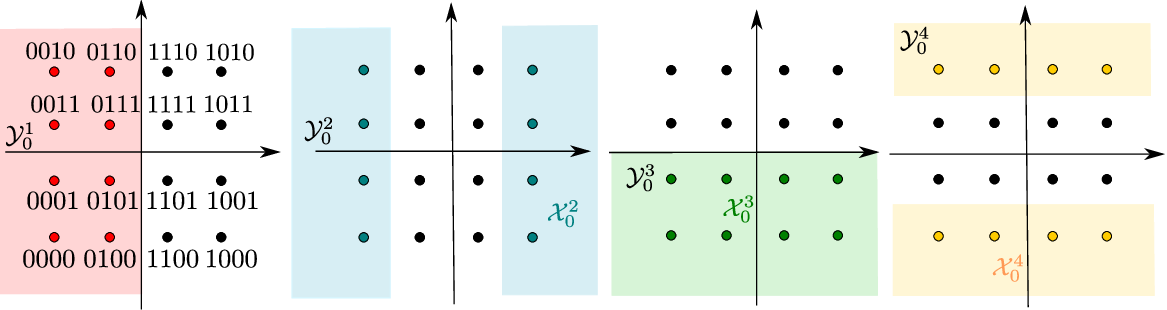}
         \caption{Decision regions of the 16-QAM constellation}
         \label{fig:16-QAM-regions}
     \end{figure}
       It can be easily seen that, for $s = 1$ and $s=3$, using the symmetries of $ \mathcal{X}^s_c$ and of the Gaussian distribution, one can write that 
     \begin{IEEEeqnarray}{rCl}
 \sum_{x \in \mathcal{X}^s_0} \int_{\mathcal{Y}_0^s} P_{Y|X} (y|x) dy = \sum_{x \in \mathcal{X}^s_1} \int_{\mathcal{Y}_1^s} P_{Y|X} (y|x) dy . 
\end{IEEEeqnarray}
However, for $s=2$ and $s=4$, the previous equality does not hold in the sense that the integration over $\mathcal{Y}_1^s$ entails a larger clipping of the Gaussian distribution than the integration over $\mathcal{Y}_0^s$. However, if the SNR is not too low, SNR $\geq 0$ dB for a normalized 16-QAM constellation, one can show that both integrations yield the same probability. The proof follows then similarly to the 8-PSK case. \qed 
 
\section{Proof of Theorem \ref{thm:sufficient_statistic}}\label{app:proof_sufficient_statistics}
Let us start with two properties of an $(n,k)$ block code. i) The pseudo-inverse of a code $p_{inv}(\cdot)$ is defined by a $k \times n$ matrix $\mathrm{A}$ such that $\mathrm{A}\bsm{c} = \bsm{u}$; ii) the matrix $\mathrm{B}=[\mathrm{H}^T, \mathrm{A}^T]$, where $\mathrm{H}$ is the parity matrix of the code, is full rank, thus, invertible.\newpage

\noindent Next, we can write for $\bsm{w}^b_u \in \{0,1\}^k$ and $\bsm{l} \in \mathds{R}^n$
\begin{IEEEeqnarray}{rCl}
P_{\bsm{W}_u^b|\bsm{L} }(\bsm{w}^b_u  | \bsm{l})   &=& P_{\bsm{W}^b_u| \, |\bsm{L}|,\bsm{L}^b }(\bsm{w}_u^b | \, | \bsm{l}|,\bsm{l}^b ) \\
&=& P_{\bsm{W}^b_u| \, |\bsm{L}|,\mathrm{H}\bsm{L}^b, \mathrm{A}\bsm{L}^b }(\bsm{w}_u^b | \, | \bsm{l}|,\mathrm{H}\bsm{l}^b, \mathrm{A}\bsm{l}^b ) , \quad  
\end{IEEEeqnarray}
where we have used the fact that $\mathrm{B}=[\mathrm{H}^T, \mathrm{A}^T]$ is invertible. Next, recalling the result of Theorem \ref{thm:binary_channel_model} in \eqref{eq:equivalent_BSC}, we have that 
\begin{equation}
    \mathrm{A}\bsm{L}^b = \mathrm{A}\bsm{C} \oplus \bsm{W}^b_u =  \bsm{U} \oplus \bsm{W}^b_u .
\end{equation} 
Since, $\bsm{U}$ is i.i.d following a Bern$(0.5)$ distribution, and since $\bsm{W}^b_u$ is independent of $\bsm{U} $, then $A\bsm{L}^b$ is Bern$(0.5)$ and is independent of $\bsm{W}^b_u$. Hence, it follows that
\begin{equation*}
     P_{\bsm{W}^b_u| |\bsm{L}|,\mathrm{H}\bsm{L}^b, \mathrm{A}\bsm{L}^b }(\bsm{w}_u^b | | \bsm{l}|,\mathrm{H}\bsm{l}^b, \mathrm{A}\bsm{l}^b )\text{=}  P_{\bsm{W}^b_u| |\bsm{L}|,\mathrm{H}\bsm{L}^b }(\bsm{w}_u^b | | \bsm{l}|,\!\mathrm{H}\bsm{l}^b )
\end{equation*} 
 which completes the proof. \qed 
 
\bibliographystyle{IEEEtran}
\bibliography{Paper_decoder}

\begin{thebibliography}{10}
\providecommand{\url}[1]{#1}
\csname url@samestyle\endcsname
\providecommand{\newblock}{\relax}
\providecommand{\bibinfo}[2]{#2}
\providecommand{\BIBentrySTDinterwordspacing}{\spaceskip=0pt\relax}
\providecommand{\BIBentryALTinterwordstretchfactor}{4}
\providecommand{\BIBentryALTinterwordspacing}{\spaceskip=\fontdimen2\font plus
\BIBentryALTinterwordstretchfactor\fontdimen3\font minus \fontdimen4\font\relax}
\providecommand{\BIBforeignlanguage}[2]{{%
\expandafter\ifx\csname l@#1\endcsname\relax
\typeout{** WARNING: IEEEtran.bst: No hyphenation pattern has been}%
\typeout{** loaded for the language `#1'. Using the pattern for}%
\typeout{** the default language instead.}%
\else
\language=\csname l@#1\endcsname
\fi
#2}}
\providecommand{\BIBdecl}{\relax}
\BIBdecl

\bibitem{Gruber_2017}
T.~Gruber, S.~Cammerer, J.~Hoydis, and S.~ten Brink, ``On {D}eep {L}earning-{B}ased {C}hannel {D}ecoding,'' in \emph{2017 51st Annual Conference on Information Sciences and Systems ({CISS})}.\hskip 1em plus 0.5em minus 0.4em\relax {IEEE}, mar 2017.

\bibitem{OShea_2017}
T.~O'Shea and J.~Hoydis, ``An {I}ntroduction to {D}eep {L}earning for the {P}hysical {L}ayer,'' \emph{{IEEE} Transactions on Cognitive Communications and Networking}, vol.~3, no.~4, pp. 563--575, dec 2017.

\bibitem{Nachmani_2016}
E.~Nachmani, Y.~Be{\textquotesingle}ery, and D.~Burshtein, ``Learning to {D}ecode {L}inear {C}odes {U}sing {D}eep {L}earning,'' in \emph{2016 54th Annual Allerton Conference on Communication, Control, and Computing (Allerton)}.\hskip 1em plus 0.5em minus 0.4em\relax {IEEE}, sep 2016.

\bibitem{Nachmani_2021}
E.~Nachmani and L.~Wolf, ``Autoregressive {B}elief {P}ropagation for {D}ecoding {B}lock {C}odes,'' 2021.

\bibitem{Xu_2017}
W.~Xu, Z.~Wu, Y.-L. Ueng, X.~You, and C.~Zhang, ``Improved {P}olar {D}ecoder {B}ased on {D}eep {L}earning,'' in \emph{2017 {IEEE} International Workshop on Signal Processing Systems ({SiPS})}.\hskip 1em plus 0.5em minus 0.4em\relax {IEEE}, oct 2017.

\bibitem{Bennatan_2018_arxiv}
A.~Bennatan, Y.~Choukroun, and P.~Kisilev, ``{D}eep {L}earning for {D}ecoding of {L}inear {C}odes - {A} {S}yndrome-{B}ased {A}pproach,'' 2018.

\bibitem{Choukroun_2022}
Y.~Choukroun and L.~Wolf, ``Error {C}orrection {C}ode {T}ransformer,'' 2022.

\bibitem{Lugosch_2018}
L.~Lugosch and W.~J. Gross, ``Learning from the {S}yndrome,'' in \emph{2018 52nd Asilomar Conference on Signals, Systems, and Computers}.\hskip 1em plus 0.5em minus 0.4em\relax {IEEE}, oct 2018.

\bibitem{Arikan_2009}
E.~Arikan, ``Channel {P}olarization: {A} {M}ethod for {C}onstructing {C}apacity-{A}chieving {C}odes for {S}ymmetric {B}inary-{I}nput {M}emoryless {C}hannels,'' \emph{{IEEE} Transactions on Information Theory}, vol.~55, no.~7, pp. 3051--3073, jul 2009.

\bibitem{Caciularu_2021}
A.~Caciularu, N.~Raviv, T.~Raviv, J.~Goldberger, and Y.~Be’ery, ``perm2vec: {A}ttentive {G}raph {P}ermutation {S}election for {D}ecoding of {E}rror {C}orrection {C}odes,'' \emph{IEEE Journal on Selected Areas in Communications}, vol.~39, no.~1, pp. 79--88, 2021.

\bibitem{DeBoni_2023}
G.~De~Boni~Rovella and M.~Benammar, ``{Improved Syndrome-based Neural Decoder for Linear Block Codes},'' in \emph{2023 IEEE Global Communications Conference}, 2023, pp. 5689--5694.

\bibitem{Caire_1998}
G.~Caire, G.~Taricco, and E.~Biglieri, ``Bit-interleaved coded modulation,'' \emph{IEEE Transactions on Information Theory}, vol.~44, no.~3, pp. 927--946, 1998.

\bibitem{Alvarado_2008}
A.~Alvarado, \emph{On bit-interleaved coded modulation with {QAM} constellations}.\hskip 1em plus 0.5em minus 0.4em\relax Chalmers Tekniska Hogskola (Sweden), 2008.

\bibitem{Cho_2014}
K.~Cho, B.~Merrienboer, C.~Gulcehre, F.~Bougares, H.~Schwenk, and Y.~Bengio, ``Learning {P}hrase {R}epresentations using {RNN} {E}ncoder-{D}ecoder for {S}tatistical {M}achine {T}ranslation,'' 06 2014.

\bibitem{MartinAbadi_2015}
\BIBentryALTinterwordspacing
M.~A. et~al., ``{TensorFlow}: {L}arge-{S}cale {M}achine {L}earning on {H}eterogeneous {S}ystems,'' 2015, software available from tensorflow.org. [Online]. Available: \url{https://www.tensorflow.org/}
\BIBentrySTDinterwordspacing

\bibitem{Chollet_2015}
\BIBentryALTinterwordspacing
F.~Chollet \emph{et~al.}, ``Keras,'' 2015. [Online]. Available: \url{https://keras.io}
\BIBentrySTDinterwordspacing

\bibitem{Kingma_2014}
D.~P. Kingma and J.~Ba, ``Adam: {A} {M}ethod for {S}tochastic {O}ptimization,'' 2014.

\bibitem{Dey_2017}
R.~Dey and F.~M. Salem, ``Gate-variants of {G}ated {R}ecurrent {U}nit ({GRU}) neural networks,'' in \emph{2017 IEEE 60th International Midwest Symposium on Circuits and Systems (MWSCAS)}, 2017, pp. 1597--1600.

\bibitem{Vaswani_2017}
A.~Vaswani \emph{et~al.}, ``Attention is {A}ll you {N}eed,'' in \emph{Advances in Neural Information Processing Systems}, vol.~30.\hskip 1em plus 0.5em minus 0.4em\relax Curran Associates, Inc., 2017.

\end{thebibliography}

\end{document}